\documentclass[12pt,reqno]{article}

\usepackage[usenames]{color}
\usepackage{amssymb}
\usepackage{graphicx}
\usepackage{amscd}

\usepackage[colorlinks=true,
linkcolor=webgreen,
filecolor=webbrown,
citecolor=webgreen]{hyperref}

\definecolor{webgreen}{rgb}{0,.5,0}
\definecolor{webbrown}{rgb}{.6,0,0}

\usepackage{color}
\usepackage{fullpage}
\usepackage{float}

\usepackage{graphics,amsmath,amssymb}
\usepackage{amsthm}
\usepackage{amsfonts}
\usepackage{latexsym}
\usepackage{epsf}

\DeclareMathOperator{\cexp}{cexp}

\def\modd#1 #2{#1\ ({\rm mod}\ #2)}

\begin{document}

\title{Minimum Critical Exponents for Palindromes}

\author{Jeffrey Shallit \\
School of Computer Science \\
University of Waterloo \\
Waterloo, ON  N2L 3G1 \\
Canada\\
{\tt shallit@cs.uwaterloo.ca}}

\maketitle

\theoremstyle{plain}
\newtheorem{theorem}{Theorem}
\newtheorem{corollary}[theorem]{Corollary}
\newtheorem{lemma}[theorem]{Lemma}
\newtheorem{proposition}[theorem]{Proposition}

\theoremstyle{definition}
\newtheorem{definition}[theorem]{Definition}
\newtheorem{example}[theorem]{Example}
\newtheorem{conjecture}[theorem]{Conjecture}

\theoremstyle{remark}
\newtheorem{remark}[theorem]{Remark}

\begin{abstract}
We determine the minimum possible critical exponent for 
all palindromes over finite alphabets.
\end{abstract}

\section{Introduction}

Palindromes --- words that read the same forwards and backwards, such as the
English word {\tt radar} --- 
have been studied in formal language theory for at least 50 years,
starting with the 1965 papers of
Freivalds \cite{Freivalds:1965} and Barzdin \cite{Barzdin:1965}.
Recently Mikhailova and Volkov \cite{Mikhailova&Volkov:2009}
initiated the study of repetitions in palindromes.

In this note we completely classify the largest fractional repetitions
in palindromes.

\section{Notation}

Throughout, we use the variables $a,b,c$ to denote single letters, and
the variables $u,v,w,x,y,z$ to denote words.   By $|x|$ we mean the
length of a word $x$, and by $x^R$ we mean its reversal.  The
empty word is written $\varepsilon$.  A word $x$ is a {\it palindrome}
if $x = x^R$.  It is well-known, and easily proved, that
$(xy)^R = y^R x^R$.

If a word $w$ can be written in the form $w = xyz$ for (possibly empty)
words $x,y,z$, then we say that $y$ is a {\it factor\/} of $w$.  
We say that a word $x = x[1..n]$ has {\it period\/} $p$ if
$x[i] = x[i+p]$ for $1 \leq i \leq n-p$.
We say that a word
$x$ is a $(p/q)$-power, for integers $p > q \geq 1$, if
$x$ has period $q$ and length $p$.  
For example, the word {\tt ionization} is a ${{10} \over 7}$-power.
A $2$-power is called
a {\it square}.
The exponent of a finite word $w$,
written $\exp(w)$, is defined to be the largest rational number
$p/q$ such that $w$ is a $(p/q)$-power.
For example,
$x = {\tt tretretretre}$, the name of an extinct Madagascan lemur,
is both a $2$-power and a $4$-power, but $\exp(x) = 4$.

Finally, we say that
a word $z$ {\it contains\/} an $\alpha$-power if $z$ contains a factor
$x$ that is a $(p/q)$-power for some $p/q \geq \alpha$.  Otherwise
we say that $z$ {\it avoids} $\alpha$-powers or is $\alpha$-power free.
We say that a word $z$ avoids $\alpha^+$-powers or
is $\alpha^+$-power-free if,
for all $p/q > \alpha$,
the word $z$ contains no factor that is a $(p/q)$-power.
The {\it critical exponent\/} of a word, written $\cexp(w)$, is the maximum
of $\exp(w')$ over all nonempty factors $w'$ of $w$.

A {\it morphism} $h$ is a map from $\Sigma^*$ to $\Delta^*$ for
alphabets $\Sigma$ and $\Delta$ satisfying the rule
$h(xy) = h(x) h(y)$ for all words $x, y \in \Sigma^*$.
It is said to be {\it $k$-uniform\/} if the image length of every letter is
equal to $k$.
If $\Delta \subseteq \Sigma$, we
can iterate $h$, writing $h^2$ for the composition $h \circ h$,
$h^3$ for $h \circ h \circ h$, and so forth.  If there is a letter
$a$ such that $h(a) = ax$, with $h(x) \not= \epsilon$,
we can iterate $h$ on $a$, obtaining the infinite {\it fixed point\/}
$$ h^\omega(a) = a \, x \, h(x) \, h^2(x) \, h^3(x) \cdots .$$
As an example, consider the {\it Thue-Morse morphism} $\mu$, defined
by $\mu(0) = 01$, $\mu(1) = 10$.  Then ${\bf t} = \mu^\omega (0)$ is
the infinite {\it Thue-Morse word}, studied by
Thue \cite{Thue:1912,Berstel:1995}.  As is well-known,
the word {\bf t} is $2^+$-power-free.

The following simple lemma will be useful for us.
We say that a morphism $h : \Sigma^* \rightarrow \Sigma^*$
is {\it palindromic\/} if every letter has a palindromic image; that is,
if $h(a)$ is a palindrome for all $a \in \Sigma$.

\begin{lemma}
Suppose $h$ is a palindromic morphism, and $w$ is a palindrome.
Then $h(w)$ is a palindrome.
\label{lem1}
\end{lemma}

\begin{proof}
Write $w = a_1 a_2 \cdots a_n$ for letters $a_i$.    Then
\begin{align*}
h(w) & = h(w^R) = h(a_n \cdots a_2 a_1) \\
&= h(a_n) \cdots h(a_2) h(a_1) = h(a_n)^R \cdots h(a_2)^R h(a_1)^R  \\
&= (h(a_1) h(a_2) \cdots h(a_n))^R = h(w)^R .
\end{align*}
\end{proof}

\section{The even-length case}

In this paper we are interested in the critical exponents of palindromes.
Even-length palindromes are not so interesting:  the symbol in the middle
of every even-length palindrome is repeated $\cdots aa \cdots$,
and so every even-length palindrome has critical
exponent at least $2$.  Over a binary alphabet (and hence over every
alphabet), it is easy to see that $\mu^{2n} (0)$ is a palindromic
prefix of the Thue-Morse word of length $2^{2n}$. Hence for
all alphabets and all even lengths the smallest possible 
critical exponent is $2$, and it is always achieved.

In the remainder of this paper, then, we focus on the odd-length palindromes.
We note that, in order to show the existence of a palindrome of 
length $\ell$ avoiding $\alpha$-powers, for all odd $\ell \geq 1$,
it suffices to exhibit arbitrarily large such palindromes.  This is
because every odd-length palindrome avoiding $\alpha$-powers continues
to avoid $\alpha$-powers if a block of size $t$ is removed simultaneously
from the front and end of the word.

\section{Binary alphabet}

In this section we assume the alphabet is $\Sigma = \{ 0,1 \}$.

\begin{proposition}
Every odd-length binary palindrome of length $\geq 7$ contains
a ${7\over 3}$-power.
\end{proposition}

\begin{proof}
Every odd-length binary palindrome of length $\geq 7$ contains within it
an odd-length palindrome of length $7$.  So it suffices to examine the
$16$ possible odd-length binary palindromes of length $7$.  The minimum
critical exponents are those corresponding to the words
$0110110$ and $1001001$, with exponent ${7 \over 3}$.
\end{proof}

\begin{theorem}
For every odd $\ell \geq 7$, there is a palindrome of length $\ell$
with critical exponent ${7 \over 3}$.
\end{theorem}

\begin{proof}
We use the following $19$-uniform
morphism of Rampersad:
\begin{align*}
f(0) &= 0110100110110010110 \\
f(1) &= 1001011001001101001 
\end{align*}
Note that $f$ is palindromic, and so $f^n(0)$ is a palindrome of
odd length for all $n \geq 1$.  Rampersad showed \cite{Rampersad:2004} that 
$f^\omega(0)$ avoids ${7\over3}^+$-powers.  Hence $f^n(0)$, for 
$n \geq 1$, gives us a sequence of longer and longer palindromes
with the desired property.
\end{proof}

\begin{remark}
It is easy to see that each of the words
$\mu^{2n} (0) \, 010 \, \mu^{2n} (0)$ is a palindrome of odd length.
It is possible to show, although we do not do it here, that
each of these words has critical
exponent ${7 \over 3}$.  In fact, the only ${7 \over 3}$-power is
$1001001$, which occurs exactly once at the center.
This gives an alternative construction.
\end{remark}

\section{Ternary alphabet}

\begin{proposition}
Every odd palindrome of length $\geq 17$ over a ternary alphabet
contains a $(7/4)$-power.
\end{proposition}

\begin{proof}
It suffices to examine all 19683 length-17 palindromes.
The word 
$$01210120102101210,$$
and its images under codings that permute
the underlying alphabet, are the
unique length-$17$ palindromes with exponent ${7 \over 4}$.
\end{proof}

\begin{theorem}
For every odd $\ell \geq 17$ there is a palindrome of length
$\ell$ with critical exponent ${7 \over 4}$.
\end{theorem}

\begin{proof}
Define the $19$-uniform morphism $g$ by
\begin{align*}
g(0) &= 0120212012102120210 \\
g(1) &= 1201020120210201021 \\
g(2) &= 2012101201021012102 \\
\end{align*}
It is easily checked that $g$ is palindromic, and so
from Lemma~\ref{lem1}, we see that $g^n(0)$ is a
palindrome of length $19^n$.    However, $g$ is --- up to renaming
of the letters --- just Dejean's
famous morphism \cite{Dejean:1972}
and she proved its iterates avoid ${7 \over 4}^+$-powers.
\end{proof}

\section{Alphabet size $4$ and larger}

In this section we assume the alphabet is
$\Sigma = \{ 0,1,2,3 \}$ or larger.

\begin{proposition}
All odd palindromes of length $\geq 3$ have a ${3 \over 2}$-power.
\end{proposition}

\begin{proof}
Such a palindrome must contain, at its center,
the word $aba$, which is a ${3 \over 2}$-power.
\end{proof}

\begin{theorem}
Over an alphabet of size $4$ or larger,
for every odd $\ell \geq 3$ there is a palindrome
of length $\ell$ with critical exponent ${3 \over 2}$.
\end{theorem}

\begin{proof}
It suffices to do this for an alphabet of size $4$.
Consider the $11$-uniform morphism $h$ defined by
\begin{align*}
h(0) &= 01312021310 \\
h(1) &= 12023132021 \\
h(2) &= 23130203132 \\
h(3) &= 30201310203 \\
\end{align*}
It is easy to check that $h$ is palindromic.
We claim that $h^\omega (0)$ is ${3 \over 2}^+$-power-free.
To see this, we use the {\tt Walnut} theorem-proving package of
H. Mousavi \cite{Mousavi:2016}.

The appropriate predicate is \\
{\tt eval abc "?msd\_11 Ei (p>=1) \& Ak Aj ((2*j<=p)\&(k=i+j)) => H[k]=H[k+p]":} \\
where {\tt H} is an automaton defined in {\tt Walnut}'s 
``Word Automata'' directory to represent the morphism $h$.  The
computation ran in 28.982 seconds on a Linux machine and proved the 
result.   It follows from Lemma~\ref{lem1} that $h^n(0)$ is a 
palindrome of length $11^n$ that is ${3 \over 2}^+$-power-free.

Furthermore, a similar computation gives that the only
${3 \over 2}$-powers occurring in $h^\omega(0)$ are
of length $a \cdot 11^i$, where $a \in {3,6,9,12}$.
\end{proof}

\section{Bi-infinite words}

For some purposes it might be nice to have a {\it single\/} bi-infinite
word $\cdots a_{-2} a_{-1} a_0 a_1 a_2  \cdots$ all of whose
truncations $a_{-n} \cdots a_{-1} a_0 a_1 \cdots a_n$ are palindromes
and avoid the appropriate powers.  For this it is useful to have
a palindromic
morphism $\gamma$ such that $\gamma(0) = x^R 0 x$ for some word $x$.
Then iterating $\gamma$ around the central $0$ gives the
bi-infinite word $\cdots \gamma^2(x^R) \gamma(x^R) 0 \gamma(x) \gamma^2(x) \cdots$.

For the binary case, Rampersad's morphism $f$ gives this, as 
$f(0) = x^R 0 x$ for $x = 110010110$.   Thus, iteration of $f$
around the central $0$ preserves the $0$ and generates such a bi-infinite 
word.  

Similarly, for alphabet size $4$, our morphism $h$ can be written
in the form $h(0) = y^R 0 y$, for $y = 21310$.  So again, iteration
of $h$ around the central $0$ preserves the $0$ and generates such
a bi-infinite word.

However, Dejean's morphism $g$ does not have this property.
However, $g^3$ does; it is $6859$-uniform.
Alternatively, one can use the  $31$-uniform morphism defined
as follows:
\begin{align*}
\alpha(0) &= 0121021201020120210201021201210 \\
\alpha(1) &= 1202102012101201021012102012021 \\
\alpha(2) &= 2010210120212012102120210120102 
\end{align*}
Then $\alpha$ is palindromic and it can be shown
that $\alpha^\omega (0)$ avoids ${7 \over 4}^+$ powers. 
Since $\alpha(0) = z^R 0 z$ for $z = 210201021201210$,
iteration of $\alpha$ around the central $0$ preserves the $0$ and
generates a bi-infinite word.

\section{Acknowledgment}

We thank Narad Rampersad for pointing out the paper of
Mikhailova and Volkov.

\section{Further work}

In the next version of this paper we will have results on the exponential
growth of the number of palindromes with minimal exponent over all
alphabets.

\newcommand{\noopsort}[1]{} \newcommand{\singleletter}[1]{#1}


\begin{thebibliography}{1}

\bibitem{Barzdin:1965}
J.~M. Barzdin.
\newblock Complexity of recognition of symmetry on {Turing} machines.
\newblock {\em Problemy Kibernetiki} {\bf 15} (1965), 245--248.
\newblock In Russian.

\bibitem{Berstel:1995}
J.~Berstel.
\newblock {\em Axel {Thue's} Papers on Repetitions in Words: a Translation}.
\newblock Number~20 in Publications du Laboratoire de Combinatoire et
  d'Informatique {Math\'ematique}. Universit\'e du Qu\'ebec \`a Montr\'eal,
  February 1995.

\bibitem{Dejean:1972}
F.~{Dejean}.
\newblock Sur un {th\'eor\`eme} de {Thue}.
\newblock {\em J. Combin. Theory. Ser. A} {\bf 13} (1972), 90--99.

\bibitem{Freivalds:1965}
R.~Freivalds.
\newblock Complexity of palindromes recognition by {Turing} machines with an
  input.
\newblock {\em Algebra i Logika} {\bf 4}(1) (1965), 47--58.
\newblock In Russian.

\bibitem{Mikhailova&Volkov:2009}
I.~A. Mikhailova and M.~V. Volkov.
\newblock Pattern avoidance by palindromes.
\newblock {\em Theoret. Comput. Sci.} {\bf 410} (2009), 2992--2998.

\bibitem{Mousavi:2016}
H.~Mousavi.
\newblock Automatic theorem proving in {{\tt Walnut}}.
\newblock Preprint available at \url{https://arxiv.org/abs/1603.06017}, 2016.

\bibitem{Rampersad:2004}
N.~Rampersad.
\newblock Words avoiding ${7 \over 3}$-powers and the {Thue-Morse} morphism.
\newblock In C.~S. Calude, E.~Calude, and M.~J. Dinneen, editors, {\em
  Developments in Language Theory, 8th International Conference, DLT 2004},
  Vol. 3340 of {\em Lecture Notes in Computer Science}, pp.  357--367.
  Springer-Verlag, 2004.

\bibitem{Thue:1912}
A.~Thue.
\newblock {\"Uber} die gegenseitige {Lage} gleicher {Teile} gewisser
  {Zeichenreihen}.
\newblock {\em Norske vid. Selsk. Skr. Mat. Nat. Kl.} {\bf 1} (1912), 1--67.
\newblock Reprinted in {\it Selected Mathematical Papers of Axel Thue}, T.
  Nagell, editor, Universitetsforlaget, Oslo, 1977, pp.~413--478.

\end{thebibliography}
\end{document}